\documentclass[letterpaper]{article}

\usepackage{times}
\usepackage{helvet}
\usepackage{courier}

\usepackage{amsthm}
\usepackage{amssymb}
\usepackage{amsmath}

\newtheorem{theorem}{Theorem}

\begin{document}

\title{Pure Price of Anarchy for Generalized Second Price Auction\thanks{This work was done when the first author and the second author were visiting Microsoft Research Asia.}}

\author{Wenkui Ding\thanks{\texttt{dingwenkui@gmail.com}. Tsinghua University, Beijing, P.R.China.} \and Tao Wu\thanks{\texttt{wutao27@mail.ustc.edu.cn}. University of Science and Technology of China, Hefei, P.R.China.} \and Tao Qin\thanks{\texttt{taoqin@microsoft.com}. Microsoft Research Asia, Beijing, P.R.China.} \and Tie-Yan Liu\thanks{\texttt{tyliu@microsoft.com}. Microsoft Research Asia, Beijing, P.R.China.}}


\date{}
\maketitle

\begin{abstract}
The Generalized Second Price auction (GSP) has been widely used by search engines to sell ad slots. Previous studies have shown that the pure Price Of Anarchy (POA) of GSP is 1.25 when there are two ad slots and 1.259 when three ad slots. For the cases with more than three ad slots, however, only some untight upper bounds of the pure POA were obtained. In this work, we improve previous results in two aspects: (1) We prove that the pure POA for GSP is 1.259 when there are four ad slots, and (2) We show that the pure POA for GSP with more than four ad slots is also 1.259 given the bidders are ranked according to a particular permutation.
\end{abstract}

\section{Introduction}
The Generalized Second Price auction (GSP), a generalization of the Vickrey auction, is the primary method used by search engines to sell online ad slots. With GSP, advertisers compete for several ad slots with decreasing click-through rates (CTR) by submitting bids to the search engine. Ad slots are assigned to advertisers according to the descending order of their bids (ad-specific CTR is also considered in practice) and each advertiser pays an amount of money if his/her ad is clicked according to the second-price rule.

GSP has drawn a lot of attention from the academic community. For example, the pure Price of Anarchy (POA) of GSP has been extensively studied. Leme and Tardos \cite{leme2010pure} exhibited the weakly feasible condition and obtained an upper bound of 1.618 for the pure POA of GSP. They also proved that the pure POA equals 1.25 when there are exactly two ad slots. Caragiannis et al. \cite{caragiannis2011efficiency} improved the upper bound (for any number of ad slots) to 1.282 based on the same condition. Lucier and Leme \cite{lucier2010improved} extended Leme and Tardos’ results and proved that the pure POA is also 1.259 when there are three ad slots. However, for the case with more than three ad slots, no tighter bounds for the pure POA of GSP have been derived, as far as we know. The motivation of this work is just to look into this issue, and advance the state of the art.

\section{Preliminaries}
Consider GSP auctions with $n$ advertisers and $n$ ad slots. Denote $\boldsymbol{v}=(v_1,v_2,\cdots v_n)$ as the advertisers' value vector, where $v_i$ is the value of the $i$-th advertiser. Denote $\boldsymbol{\alpha}=(\alpha_1,\alpha_2,\cdots\alpha_n)$ as the click-through rate (CTR) of ad slots, where $\alpha_j$ corresponds to the CTR of the $j$-th slot. Without loss of generality, we assume $1=v_1\geq v_2\geq\cdots \geq v_n\geq0$ and $1=\alpha_1\geq \alpha_2\geq \cdots \geq \alpha_n\geq0$. Given bid vector $\boldsymbol{b}=(b_1,b_2,\cdots b_n)$, we can obtain a permutation $\pi=(\pi_1,\pi_2,\cdots,\pi_n)$ of advertisers according to the descending order of their bids, where $\pi(k)$ is the index of the advertiser with the $k$-th highest bid. Following the common practice \cite{caragiannis2011efficiency,lucier2010improved,leme2010pure}, all advertisers are assumed to be conservative, i.e., $\forall i, b_i \leq v_i$.

Given the above notations, the pure POA of GSP with $n$ ad slots can be defined as follows,
\[
POA_n=\max_{\boldsymbol{v},\boldsymbol{\alpha}}\max_{\boldsymbol{b}\in \mathcal{N}(\boldsymbol{v},\boldsymbol{\alpha})} \frac{\sum_{i=1}^n \alpha_i v_i}{\sum_{i=1}^n \alpha_i v_{\pi(i)}},
\]
where $\mathcal{N}(\boldsymbol{v},\boldsymbol{\alpha})$ is the set of the pure Nash equilibrium bid vectors given $\boldsymbol{v}$ and $\boldsymbol{\alpha}$. We further define the pure POA of a given permutation $\pi$ with $n$ ad slots as:
$$POA_n(\pi)=\max_{\boldsymbol{v},\boldsymbol{\alpha}}\max_{\boldsymbol{b}\in \mathcal{N}(\boldsymbol{v},\boldsymbol{\alpha}, \pi)} \frac{\sum_{i=1}^n \alpha_i v_i}{\sum_{i=1}^n \alpha_i v_{\pi(i)}},
$$
where $\mathcal{N}(\boldsymbol{v},\boldsymbol{\alpha}, \pi)$ is the set of the pure Nash equilibrium bid vectors admitting the permutation $\pi$ given $\boldsymbol{v}$ and $\boldsymbol{\alpha}$. It is not difficult to get $POA_n=\max_{\pi\in \Omega_n}POA_n(\pi)$, where $\Omega_n$ is the set of permutations of length $n$ admitting pure Nash Equilibria.

\section{Pure POA of GSP when $n=4$}
The following theorem shows that when $n=4$, the pure POA of GSP equals 1.259.
\begin{theorem}
$POA_4=1.259$.
\end{theorem}
\begin{proof}
First, we prove $POA_4(\pi)\leq1.259$ by enumerating all possible $\pi$'s. Here we only give the details for one permutation $\pi=(2,3,1,4)$. The proofs for other permutations are just similar.

For the pure POA of the permutation $\pi=(2,3,1,4)$, we have
$
POA_4(\pi)\leq\max_{\boldsymbol{v},\boldsymbol{\alpha}}f(\boldsymbol{v},\boldsymbol{\alpha}, \pi)$, where $f(\boldsymbol{v},\boldsymbol{\alpha},\pi)=\frac{\alpha_1v_1+\alpha_2v_2+\alpha_3v_3+\alpha_4v_4}{\alpha_1v_2+\alpha_2v_3+\alpha_3v_1+\alpha_4v_4}$.
It can be verified that $\frac{\partial f(\boldsymbol{v},\boldsymbol{\alpha})}{\partial v_4}<0$. So
$f(\boldsymbol{v},\boldsymbol{\alpha})\leq f(\boldsymbol{v},\boldsymbol{\alpha})|_{v_4=0}$ and $
POA_4(\pi)\leq \max_{\boldsymbol{v},\boldsymbol{\alpha}}\frac{\alpha_1v_1+\alpha_2v_2+\alpha_3v_3}{\alpha_1v_2+\alpha_3v_1+\alpha_3v_3}.
$
Since $v_4=0$ now, we have $b_4=0$, and then $(b_1,b_2,b_3)$ is a Nash equilibrium of the game restricted to advertiser $1,2$ and $3$. The right-hand-side of the above equation should be no larger than $POA_3$ which is 1.259 \cite{caragiannis2011efficiency,lucier2010improved}.

Second, we can use the following example to show $POA_4(\pi)\geq1.259$, and therefore prove $POA_4(\pi)=1.259$. Suppose $\alpha_1=1, \alpha_2=0.55, \alpha_3=\alpha_4=0.47, v_1=1, v_2=0.53, v_3=0.15, v_4=0, b_1=0, b_2=0.53, b_3=0.15, b_4=0$. It is easy to verify that the bid vector together with the value and CTR vector compose a Nash equilibrium and its efficiency is 1.259.
\end{proof}

\emph{Remark}: Please note that the proof technique we use here is different from those used by previous works \cite{leme2010pure,caragiannis2011efficiency,lucier2010improved}. The previous works basically leverage the weakly feasible condition \cite{leme2010pure} to upper bound the pure POA: instead of directly finding the worst efficiency of all Nash equilibrium permutations, they use the worst efficiency of all permutations satisfying the weakly feasible condition as a bound. The following example shows that the upper bound obtained by using the weakly feasible condition as done in these previous works cannot be tight for $n=4$.\\
\textbf{Example}
Assume there are 4 advertisers in a GSP auction. Consider value vector $(1.00,0.53,0.25,0.16)$ and  slot CTR vector $(1.00,0.57,0.47,0.19)$. Permutation $\pi=(2,3,1,4)$ satisfies the weakly feasible condition, and its efficiency 1.269. Therefore, the upper bound obtained from the weakly feasible condition is at least 1.269. Note that we have proved $POA_4=1.259$, thus this upper bound cannot be tight.

\section{Pure POA for $n>4$}
Although we can prove $POA_4=1.259$ by enumerating all possible Nash equilibrium permutations, this approach is difficult to generalize to the cases with more ad slots. This is because there are simply too many permutations to enumerate.

\cite{lucier2010improved} provided a possible approach to prove the pure POA for any $n\geq 3$ based on a conjecture and a lemma: The
conjecture states that the permutation $\pi_n=(2,3,...,n,1)$ is the worst one in terms of pure POA among the permutations of length $n$: $POA_n(\pi_n)\geq POA_n(\sigma_n), \forall \sigma_n\in \Omega_n$ and the lemma states that $POA_n(\pi_n)=1.259$ for $POA_n(\pi_n)\leq1.259$ and $n\geq 3$.  However, we find that their proof of the result is non-rigorous (if not mistaken): In particular, the key inequality used in their proofs, $\frac{x+av}{y+bv}\leq\frac{x+av'}{y+bv'}$ given $a\leq b$ and $v\geq v'$, is incorrect; to guarantee the correctness of the inequality, a further condition $\frac{x}{y}\geq \frac{a}{b}$ is required. In this paper, we give a rigorous proof to this lemma, which is based on very different techniques from those used in \cite{lucier2010improved}.

\begin{theorem}\label{theo:per}
Given a permutation $\pi_n=(2,3,...,n,1)$ where $n\geq3$, we have $POA_n(\pi_n)\leq1.259$, and the bound is tight.
\end{theorem}
\begin{proof}
We prove the result for $n\geq5$ by induction and contradiction, since the result has been obtained for $n=3$ and $n=4$ (see Theorem 1). Assume that the result holds for $3\leq n \leq k-1$, and by contradiction $POA_k(\pi_k)>1.259$.

Let $f(\alpha,v,\pi_k)=\frac{\alpha_1v_1+\alpha_2v_2+\cdots+\alpha_kv_k}{\alpha_1v_2+\alpha_2v_3+\cdots+\alpha_kv_1}$. Since $f(\alpha,v,\pi_k)>1$, it can be verified that $\forall 2\leq i\leq k, \frac{\partial f(\alpha,v,\pi_k)}{\partial v_i}<0$. The constraints that $v_2$ must satisfy are $v_2\geq v_3$ and $v_2\geq(1-\frac{\alpha_k}{\alpha_1})v_1$. To reach the maximal $POA_k(\pi_k)$, either $v_2=v_3$ or $v_2=(1-\frac{\alpha_k}{\alpha_1})v_1$. There are three possible cases for the relationship between $v_2$ and $v_3$.

$\boldsymbol{Case\ 1}$: $v_2=v_3>(1-\frac{\alpha_k}{\alpha_1})v_1$

Because $(1-\frac{\alpha_k}{\alpha_2})v_1\leq(1-\frac{\alpha_k}{\alpha_1})v_1<v_3$, we have $v_3=v_4$. Similarly we get $v_2=v_3=\cdots=v_k=(1-\frac{\alpha_k}{\alpha_{k-1}})v_1\leq(1-\frac{\alpha_k}{\alpha_1})v_1$, which leads to a contradiction.

$\boldsymbol{Case\ 2}$: $v_2=v_3=(1-\frac{\alpha_k}{\alpha_1})v_1$

If $\alpha_1\neq\alpha_2$, we have $v_3=v_4=\cdots=v_k=(1-\frac{\alpha_k}{\alpha_{k-1}})v_1<(1-\frac{\alpha_k}{\alpha_1})v_1$, which leads to a contradiction.

If $\alpha_1=\alpha_2$, we can eliminate slot $\alpha_2$ and bidder $v_2$ simultaneously and obtain $POA_k(\pi_k)\leq POA_{k-1}(\pi_{k-1})$, which leads to a contradiction.

$\boldsymbol{Case\ 3}$: $v_2>v_3$

Now we need to consider $v_3$ and $v_4$. Similarly, there are three possible cases for the relationship between $v_3$ and $v_4$; one can discuss the three cases seperately, and for the third case one needs to dig into the relationship between $v_4$ and $v_5$ recursively.
Doing so we have
\begin{equation}
f(\alpha,v,\pi_k)=\frac{\sum_{i=1}^{k}\alpha_i-\alpha_k\sum_{i=1}^{k-1}\frac{\alpha_{i+1}}{\alpha_i}}
{\sum_{i=1}^{k}\alpha_i-(k-1)\alpha_k}.
\end{equation}

By standard techniques, we get (1) $
f(\alpha,v,\pi_k)\leq\frac{\sum_{i=1}^{k-1}\frac{\alpha_{i+1}}{\alpha_i}+\frac{\alpha_k}{\alpha_{k-1}}-1}{k-2} \leq \frac{k-1}{k-2}\lambda$
and (2)
$ f(\alpha,v,\pi_k)\leq \frac{k-(k-1)\lambda}{k-(k-1)}=k-(k-1)\lambda$, where $\lambda=\frac{1}{k-1}\sum_{i=1}^{k-1}\frac{\alpha_{i+1}}{\alpha_i}$ and $0<\lambda<1$.
Thus,
$
POA_k(\pi_k)\leq \max_{\alpha,v}f(\alpha,v,\pi_k)\leq\min\{\frac{k-1}{k-2}\lambda,k-(k-1)\lambda\}\leq1.25, \forall n\geq5
$,
which leads to another contradiction.

Combing the above three cases, we proved that $POA_n(\pi_n)\leq 1.259$. Further, we can construct an example to show the tightness of the bound by extending the example used in the proof of Theorem 1: $\forall 4\leq k \leq n, \alpha_k=0.47, v_k=0, b_k=0$.
\end{proof}

\bibliography{POA}
\bibliographystyle{plain}

\end{document}